\documentclass[9pt]{article}
\usepackage{spconf,amsmath,epsfig,bm,algorithm,algorithmic,float}
\usepackage{graphicx}
\usepackage{epsfig}
\usepackage{amsmath}
\usepackage{amsmath,epsfig,bm,algorithm,algorithmic,float}
\usepackage{amstext,graphicx,amssymb,amsfonts,xcolor,enumerate}
\usepackage{mathrsfs,cancel}
\usepackage{graphicx}
\usepackage{epstopdf}
\usepackage{amsmath,bm,algorithm,algorithmic,float}
\usepackage{amssymb}
\usepackage[thmmarks,amsmath,amsthm,hyperref]{ntheorem}
\usepackage{multirow}
\usepackage{slashbox, color}
\def\lc{\left\lceil}
\def\rc{\right\rceil}
\newtheorem{theorem}{Theorem}

\title{Fast Template Matching by Subsampled Circulant Matrix}
\name{}
\name{Sung-Hsien Hsieh$^{*,**}$, Chun-Shien Lu$^{*}$, and Soo-Chang Pei$^{**}$}
\address{$^{*}$Institute of Information Science, Academia Sinica, Taipei, Taiwan\\
$^{**}$Graduate Inst. Comm. Eng., National Taiwan University, Taipei, Taiwan}
\begin{document}

\maketitle

\begin{abstract}
Template matching is widely used for many applications in image and signal processing and usually is time-critical. Traditional methods usually focus on how to reduce the search locations by coarse-to-fine strategy or full search combined with pruning strategy. However, the computation cost of those methods is easily dominated by the size of signal $N$ instead of that of template $K$. This paper proposes a probabilistic and fast matching scheme, which computation costs requires $O(N)$ additions and $O(K \log K)$ multiplications, based on cross-correlation. The nuclear idea is to first downsample signal, which size becomes $O(K)$, and then subsequent operations only involves downsampled signals. The probability of successful match depends on cross-correlation between signal and the template. We show the sufficient condition for successful match and prove that the probability is high for binary signals with $\frac{K^{2}}{\log K} \geq O(N)$. The experiments shows this proposed scheme is fast and efficient and supports the theoretical results.
\end{abstract}

\begin{keywords}
Circulant matrix, Cross-correlation, Subsampling, Template matching
\end{keywords}

\section{Introduction}\label{sec:intro}
Template matching is a task of searching a given template in a given signal.
It has been widely used in many applications, including communication synchronization \cite{Hassanieh2012acm}, quality control \cite{Aksoy2004}, compression \cite{Luczak1997}, object detection \cite{Dufour2002}, and etc.
Given a template $\bm{t} \in \mathbb{R}^{K}$ and a signal $\bm{x} \in \mathbb{R}^{N}$, template matching is to solve
\begin{equation}
\label{eq: TM}
\displaystyle \arg\max_{0<k < N}\ Sim\left(\bm{x}^{k},t\right),
\end{equation}
where $\bm{x}^{k}=\left[ (\bm{x})_{k},...,(\bm{x})_{N-1},(\bm{x})_{0},...,(\bm{x})_{k-1} \right]$, $(\bm{x})_{k}$ represents the $k^{th}$ entry of $\bm{x}$, and $Sim\left( \cdot \right)$ denotes a kind of similarity metric.

One main challenge of solving Eq. (\ref{eq: TM}) is to incur high computation time but many applications demands real-time and are energy-critical.
There are two factors influencing computation overheads.
One is what similarity measures, including cross-correlation (CC) \cite{Anuta1970}\cite{Lewis1995}, normalized cross-correlation (NCC) \cite{wei2008}\cite{Pan2008}, and sum of squared differences (SSD)\cite{Santini1999}, are used.
Similarity measure influences performances as well.
Another is the search strategy.
For example, exhaustive search (full search) that is conducted by checking $\bm{x}^{k}$ from $k=0$ to $k=N-1$, which is unacceptably slow.

Many techniques have been developed to overcome these difficulties as follows.
\begin{enumerate}
\item[$\bullet$] Coarse-to-fine strategy \cite{VanderBrug1977}\cite{Gharavi2001}\cite{Mahmood2012}: First, a coarse search is conducted by finding the downsampled template in the downsampled image to yield a good match with less computation overhead.
Then, a fine search is conducted in the original space starting from the neighborhood of the best match found in coarse search.
\item[$\bullet$] Full search-equivalent strategy \cite{Mattoccia2008}\cite{Schweitzer2011}\cite{Quyang2012}: It employ rejection schemes derived, where current search is terminated as soon as some criterions are satisfied.
\item[$\bullet$] Fast convolution by Fast Fourier Transform (FFT) \cite{Hassanieh2012acm}\cite{Lewis1995}: For specific similarity measures such CC and SSD, Eq. (\ref{eq: TM}) requires convoluting $\bm{x}$ with $\bm{t}$.
By convolution theorem, this operation can be quickly done in the frequency domains of $\bm{x}$ and $\bm{t}$.
In particular, \cite{Hassanieh2012acm} proposing using sparse FFT \cite{Haitham2012} to replace FFT for furhter reducing the computation cost.
\end{enumerate}

In this paper, we propose a fast template matching method using downsampled circulant matrix and cross-correlation.
For image-based applications, the performance using cross-correlation is usually inferior to that using NCC or SSD.
But, for binary signals such as CDMA codes in communications, it indeed works well \cite{Hassanieh2012acm}.
Our idea is to search the match for downsampled signal, and no search in the higher resolution space is required. It results in the fact that the computation overhead, which is far less than other approaches, is related to the size of downsampled signal.
The crux of our method relies on how to design a circulant matrix to achieve fast template matching and can be considered to be an extension of \cite{Hassanieh2012acm}.
The main diferences are two-fold:
(1) Our idea is based on exploiting the commutative property between circulant matrices instead of adopting sFFT, leading to the advantages that the proposed scheme is more simple and faster than \cite{Hassanieh2012acm}.
(2) The computation cost of our method benefits from the size of template, but that of \cite{Hassanieh2012acm} only depends on the size of signal.

\section{Notations}\label{sec:notations}
We briefly introduce the notations used in this paper.
Let $\bm{v} \in \mathbb{R}^{N}$ be a vector and let $circ(\bm{v})$ be a circulant matrix generated based on the seed vector $\bm{v}$.
For example, $\bm{V}=circ(\bm{v})$, where the first row of $\bm{V}$ is $[(\bm{v})_{0},(\bm{v})_{1},...,(\bm{v})_{N-1} ]$, the second row is $[(\bm{v})_{N-1},(\bm{v})_{0},...,(\bm{v})_{N-2} ]$, and the last row is $[(\bm{v})_{1},(\bm{v})_{2},...,(\bm{v})_{0} ]$.
$0_{1 \times (N)}$ represents the $1 \times N$ zero vector.
If $\bm{u} \in \mathbb{R}^{M}$ and $\bm{v} \in \mathbb{R}^{N}$, $[\bm{u}\ \bm{v}] \in \mathbb{R}^{M+N}$ denotes a vector by concatenating $\bm{u}$ and $\bm{v}$.
We simplify notations to use $(\bm{x})_{i}=(\bm{x})_{i\ modulo \ N}$.

\section{Proposed Method}\label{sec:proposed approach}
Template matching based on cross-correlation is equivalent to solving $\displaystyle m = \arg\max_{k}\ (\bm{Tx})_{k} $ with $\displaystyle (\bm{Tx})_{k}= \sum_{i=0}^{K-1} (\bm{t})_{i}(\bm{x})_{k+i}$, where $\bm{T}=circ([\bm{t}\ 0_{1 \times (N-K)}])$.
Specifically, $(\bm{Tx})_{k}$ is considered as the matching result at position $k$ of $\bm{x}$ and $m$ is considered as the ground truth with the maximum correlation, $(\bm{Tx})_{m}$.
The goal of the proposed method is to find an index $\hat{m}$, which is expected to be equal to the ground truth $m$, with fast computation.

The key idea is to downsample an original signal $\bm{x}$ via a sampling matrix $\Phi$ to become a low-dimensional signal $\bm{y}$, {\em i.e.},  $\bm{y}=\Phi \bm{x}$, and
then find a matrix $\hat{\bm{T}} \in \mathbb{R}^{M \times M}$ such that $\hat{\bm{T}}\bm{y} = \hat{\bm{T}}\Phi \bm{x}  = \Phi \bm{T x}$.
Since,  compared with $\bm{x}$, $\bm{y}$ has lower dimension, the computation cost of $\hat{\bm{Ty}}$ is lower than that of $\bm{Tx}$.
In other words, $\hat{\bm{Ty}}$ is equivalent to fast convolving $\bm{x}$ with $\bm{T}$.
Moreover, $\Phi \bm{T x}$ is considered to downsample $\bm{Tx}$, which is a vector representing matching results based on cross-correlation.
Nevertheless, downsampling ({\em i.e.}, for $\Phi \bm{T x}$) also leads to the side effect that the matching result with the maximum cross-correlation cannot be identified intuitively from $\hat{\bm{Ty}}$.
To overcome the difficulty, we first require at least two downsampled signals for conducting template matching and then by Chinese Remainder Theorem (CRT) the ground truth can be correctly identified.

The proposed algorithm is depicted in Algorithm \ref{alg:FTM}, which is mainly composed of four operations:
\begin{enumerate}
\item[1.] Downsampling: Downsample original signal $\bm{x}$ into low-dimensional signal $\bm{y}$, where each entry of $\bm{y}$ is the sum of $\lc \frac{N}{M} \rc$ entries in $\bm{x}$ (Step 2-3).
\item[2.] Fast convolution in low-dimensional space: Convolve $\bm{y}$ with $\bm{t}$ and obtain the result $\bm{r}$ (Step 4-6).
By commutative property between two circulant matrices, $\bm{y}$ being convolved by $\bm{t}$ is equivalent to $\bm{x}$ being convolved by $\bm{t}$. Each entry of the convolution result $\bm{r}$ is transformed as the sum of $\lc \frac{N}{M} \rc$ entries in $\bm{Tx}$.
\item[3.] Matching position finding: If an entry in $\bm{r}$ is large, it implies that one of $\lc \frac{N}{M} \rc$ entries in $\bm{Tx}$ is also large with high probability. Thus, by searching the maximum value in $\bm{r}$ (Step 7), it provides the information, a set of $\lc \frac{N}{M} \rc$ candidate positions, including $m$.
\item[4.] Best matching position choice: Based on the Chinese Remainder Theorem, we can further identify the best template matching result as the unique position $\hat{m}$ (Step 8-9).
\end{enumerate}

\begin{algorithm}[!t]
\small
\centering
\setlength{\abovecaptionskip}{0pt}
\setlength{\belowcaptionskip}{0pt}
\caption{Fast Template Matching}
\label{alg:FTM}
\begin{tabular}[t]{p{8.4cm}l}
\hline
\textbf{Input:} $\bm{x} \in \mathbb{R}^{N}$, $\bm{t} \in \mathbb{R}^{K}$;\quad \textbf{Output:} $\hat{m}$;\\
\hline
01. \textbf{function} \textbf{FTM}()\\
02. \quad Pick two co-prime integers $M_{1}$ and $M_{2}$ as downsampling \\ \quad\quad\quad factors such that $M_{1}M_{2}>N$, $M_{1}\geq K$, and $M_{2}\geq K$;\\
03. \quad Downsample $\bm{x}$ into $\bm{y}_{M_{1}}$ and $\bm{y}_{M_{2}}$, where\\
\quad \ \ \quad $\displaystyle (\bm{y}_{M_{1}})_{k} = \sum_{i=0}^{\lc \frac{N}{M_{1}} \rc -1 } (\bm{x})_{k+iM_{1}} $ for $k=0,...,M_{1}+K-1$, \\
\quad \ \ \quad $\displaystyle (\bm{y}_{M_{2}})_{k} = \sum_{i=0}^{\lc \frac{N}{M_{2}} \rc -1 } (\bm{x})_{k+iM_{2}} $  for $k=0,...,M_{2}+K-1$ ; \\
04. \quad Assign $\bm{t}_{M_{1}}=[\bm{t}\ 0_{1 \times M_{1}}]$ and $\bm{t}_{M_{2}}=[\bm{t}\ 0_{1 \times M_{2}}]$;\\
05. \quad $\tilde{\bm{y}}_{M_{1}}=FFT(\bm{y}_{M_{1}})$, and   $\tilde{\bm{y}}_{M_{2}}=FFT(\bm{y}_{M_{2}})$;\\
\quad \ \ \quad $\tilde{\bm{t}}_{M_{1}}=FFT(t_{M_{1}})$, and   $\tilde{\bm{t}}_{M_{2}}=FFT(\bm{t}_{M_{2}})$;\\
06. \quad $\bm{r}_{M_{1}}=IFFT(\tilde{\bm{t}}_{M_{1}}.* \tilde{\bm{y}}_{M_{1}} )$, $\bm{r}_{M_{2}}=IFFT(\tilde{\bm{t}}_{M_{2}}.* \tilde{\bm{y}}_{M_{2}} )$; \\
\quad\quad\quad where ``.*'' denotes pixel-wise multiplication.\\
07. \quad  $\displaystyle \tilde{m}_{1} = \arg\max_{i \leq M_{1}} (\bm{r}_{M_{1}})_{i} $, $\displaystyle \tilde{m}_{2} = \arg\max_{i \leq M_{2} } (\bm{r}_{M_{2}})_{i} $;\\
08. \quad  $\bm{S}_{M_{1}}=\left\{ k | k = \tilde{m}_{1}+i*M_{1} \text{ for }  i=0,...,\lc \frac{N}{M_{1}} \rc -1   \right\}$,\\
\quad \ \ \quad $\bm{S}_{M_{2}}=\left\{ k | k = \tilde{m}_{2}+i*M_{2} \text{ for }  i=0,...,\lc \frac{N}{M_{2}} \rc -1   \right\}$;\\
09. \quad $\hat{m}= \bm{S}_{M_{1}} \bigcap \bm{S}_{M_{2}}$;\\
10. \textbf{end} \textbf{function}\\
\hline
\end{tabular}
\end{algorithm}

The main operations are discussed in detail as follows.
We will use $M$ to denote the downsampling factor if there is no confusion.

\subsection{Downsampling}\label{ssec: step 1}
Let $\bm{y}$ be a downsampled signal defined as:
\begin{equation}
\label{eq: downsample}
 (\bm{y})_{k}=\sum_{i=0}^{\lc \frac{N}{M} \rc -1 } (\bm{x})_{k+iM}= (\Phi \bm{x})_{k}
\end{equation}
with
\begin{equation}
\label{eq: Phi}
\Phi=\mathcal{D}_{M}(circ(r)),
\end{equation}
where $\mathcal{D}_{M}(\cdot)$ is the function that outputs the first $M$ rows of its argument and $\bm{r}$ is generated as follows: $(\bm{r})_{k}=1$ if $ k \text{ mod } M =0 $ for $0\leq k\leq N-1$.

There are two useful properties about $\Phi$: 1) $\Phi \bm{x}$ only involves additions and 2) $\Phi$ is a circulant matrix.
It should be noted that $M$ controls the dimension of a downsampled signal.
Furthermore, the proposed method requires $M\geq K$; otherwise, there will be no matrix satisfying $\hat{\bm{T}}\bm{y}=\Phi \bm{Tx}$ (details will be discussed later in Theorem \ref{thm: circulant matrix exchange} and Theorem \ref{thm: circulant matrix exchange 2nd Ver}).

\subsection{Fat convolution in low-dimensional space}\label{ssec: step 2}
After downsampling, we find a matrix $\hat{\bm{T}} \in \mathbb{R}^{M \times M}$ such that  $\hat{\bm{T}}\Phi \bm{x}  = \Phi \bm{T x}$.
If this strategy is feasible, it reduces the computation cost since the dimension of $\hat{\bm{T}}$ is smaller than that of $\bm{T}$. However, the problem is not intuitive but difficult in that it is not simply related to the commutative property in matrices.
Valsesia and Magli \cite{Valsesia2012} propose Theorem \ref{thm: circulant matrix exchange}, stating under what sufficient conditions the commutative property holds.

\begin{theorem} \label{thm: circulant matrix exchange} (slightly revised from \cite{Valsesia2012})
Let $\bm{T}=circ([\bm{t}\ 0_{1 \times (N-K)}]) \in \mathbb{R}^{N \times N}$ be a circulant matrix, where $\bm{t}$ is a $1 \times K$ ($K \leq M$) non-zero vector.
Let $\hat{\bm{T}} \in \mathbb{R}^{M \times M} = circ([\bm{t}\ 0_{1 \times (M-K)}])$, let $\Phi$ be a $M \times N$ partial circulant matrix, and let $\bm{y}=\Phi \bm{x}$ and $\Phi \bm{T x}$ denote the measurements of $\bm{x}$ and measurements of filtered signal ($\bm{Tx}$), respectively.
Then,
$$ \left( \hat{\bm{T}}\bm{y} \right)_{i} = \left( \Phi \bm{T x} \right)_{i}\ \ \text{if and only if} \ \ i \in [0,\ M-K ]. $$
\end{theorem}

In fact, Theorem \ref{thm: circulant matrix exchange} also implies that there are $K$ mismatches, namely $(\hat{\bm{T}} \bm{y})_{i} \neq (\hat{\bm{T}} \Phi \bm{x})_{i}$ for $M-K < i\leq M-1$.
These mismatches may lead to the failure of the proposed algorithm.
For example, $(\hat{\bm{T}} \bm{y})_{k}$, in fact, is the sum of $\lc \frac{N}{M} \rc$ entries of $\bm{Tx}$.
If $k$ is the best match position involving $(\bm{Tx})_{m}$, and $(\hat{\bm{T}} \bm{y})_{k} = (\Phi \bm{T  x})_{k}$ (without mismatches),  it means $\displaystyle (\hat{\bm{T}} \bm{y})_{k}= (\bm{Tx})_{m}+\bm{\eta}$, where $\displaystyle \bm{\eta} = \sum_{i = 0,i \neq \hat{i}}^{ \lc \frac{N}{M} \rc -1} (\bm{Tx})_{k+iM}$ and $m=k+\hat{i}M$.
We can anticipate that $(\hat{\bm{T}} \bm{y})_{k}$ may be large enough since it involves $(\bm{Tx})_{m}$.
On the contrary, if $(\hat{\bm{T}} \bm{y})_{k} \neq (\Phi \bm{T  x})_{k}$ (with mismatches), $(\hat{\bm{T}} \bm{y})_{k}$ is unpredictable since  $\displaystyle (\hat{\bm{T}} \bm{y})_{k}=(\bm{Tx})_{m}+\bm{\eta}$ no longer holds.
Under the circumstance, the proposed algorithm loses the information about $(\bm{Tx})_{m}$.

To deal with this problem, we consider two strategies: 1) Instead of using $\mathcal{D}_{M}$, we use $\mathcal{D}_{M+K}$ in Eq. (\ref{eq: Phi}) to produce $\Phi$.
In other words, according to Theorem \ref{thm: circulant matrix exchange}, $(\hat{\bm{T}} \bm{y})_{k} = (\Phi \bm{T x})_{k} $ holds for $0\leq k \leq M-1$.
It should be noted that the set of candidate positions collected from $(\hat{\bm{T}} \bm{y})_{k}$ for $0\leq k\leq M-1$ already includes the information of $(\bm{Tx})_{0},...,(\bm{Tx})_{N-1}$.
Thus, no information is lost.
However, it results in more overheads since the size of $\bm{y}$ is increased to $M+K$.
2) To further reduce such overheads, we derive another sufficient condition without mismatch by designing a new sampling matrix $\Phi$.

\begin{theorem} \label{thm: circulant matrix exchange 2nd Ver}
Suppose $\frac{N}{M}$ is an integer.
Let $\bm{T}=circ([\bm{t}\ 0_{1 \times (N-M)}]) \in \mathbb{R}^{N \times N}$ be a circulant matrix, where $\bm{t}$ is a $1 \times M$ non-zero vector, and let $\hat{\bm{T}} \in \mathbb{R}^{M \times M} = circ(\bm{t})$.
We also let $\Phi = [\underbrace{\Phi_{M}\ \Phi_{M}...\Phi_{M}}_{\frac{N}{M}}]$ be an $M \times N$ matrix, where $\Phi_{M}$ is an $M \times M$ circulant matrix.
Then, we have the measurements $\bm{y}=\Phi \bm{x}$ and measurements $\Phi \bm{T x}$ of filtered signal $\bm{Tx}$.
Therefore, we have
$$ \left( \hat{\bm{T}}\bm{y} \right)_{i} = \left( \Phi \bm{T x} \right)_{i}\ \ \text{if and only if} \ \ i \in [0,\ M-1]. $$
\end{theorem}
\begin{proof}
If both $\bm{A} \in \mathbb{R}^{N \times N}$ and $\bm{B} \in \mathbb{R}^{N \times N}$ are circulant matrices, then matrix multiplication is commutative, namely $\bm{AB}=\bm{BA}$. In our case, $\hat{\bm{T}}\bm{y} = \hat{\bm{T}}\Phi \bm{x} = \hat{\bm{T}} [\Phi_{M}\ \Phi_{M}...\Phi_{M}]\bm{x}  = [\left(\hat{\bm{T}}\Phi_{M}\right)\ \left(\hat{\bm{T}}\Phi_{M}\right)...\left(\hat{\bm{T}}\Phi_{M}\right)]\bm{x} $. Since both $\Phi_{M}$ and $\hat{\bm{T}}$ are circulant matrices, $\hat{\bm{T}}\Phi_{M}=\Phi_{M}\hat{\bm{T}}$. Thus, $$ \hat{\bm{T}}\Phi\bm{x} = [\left(\Phi_{M}\hat{\bm{T}}\right)\ \left(\Phi_{M}\hat{\bm{T}}\right)...\left(\Phi_{M}\hat{\bm{T}}\right)]x = \Phi \bar{\bm{T}}\bm{x}, $$
where $\bar{\bm{T}} = \left[ \begin{array}{cccc}
\hat{\bm{T}} & 0 & \cdots & 0 \\
 0 & \hat{\bm{T}} & \cdots &0 \\
 \vdots & \vdots & \ddots & \vdots  \\
0 & 0 & ... & \hat{\bm{T}}
\end{array} \right] $.
\end{proof}

Theorem \ref{thm: circulant matrix exchange 2nd Ver} holds when $M$ divides $N$. Otherwise, we can pad zeros into the tail of $\bm{x}$ until $M$ divides $N$.

\subsection{Matching position finding}\label{ssec: step 3}
Let $\tilde{\bm{y}}$ be the convolved signal with $\tilde{\bm{y}}=\hat{\bm{T}}\bm{y}$.
By Theorem \ref{thm: circulant matrix exchange} or Theorem \ref{thm: circulant matrix exchange 2nd Ver}, $\tilde{\bm{y}}$ can be rewritten as:
\begin{equation}
\label{eq: downsample}
 (\tilde{\bm{y}})_{k}=\sum_{i=0}^{\lc \frac{N}{M} \rc -1 } (\bm{Tx})_{k+iM}.
\end{equation}
If $\displaystyle \tilde{m} = \arg\max_{k < M}(\tilde{\bm{y}})_{k}$, we are interested in the question that whether the ground truth $m$ belonging to the candidate position set $\mathbb{C}$, which is $\mathbb{C}=\left\{ k |  k = \tilde{m}+i*M \text{ for }  i=0,...,\lc \frac{N}{M} \rc -1  \right\}$.
If yes, the proposed algorithms correctly finds $\mathbb{C}$ to include $m$.
We examine the sufficient condition of successful trials in Theorem \ref{thm: sufficient condition for perfect match}.

\begin{theorem} \label{thm: sufficient condition for perfect match}
Let $\displaystyle (\bm{Tx})_{k}= \sum_{i=0}^{K-1} (\bm{t})_{i}(\bm{x})_{k+i} $.
Let the desired matching result be $\displaystyle m = \arg\max_{k}\ (\bm{Tx})_{k}$.
If $\frac{1}{ 2\lc \frac{N}{M} \rc +1 }(\bm{Tx})_{m} > \max_{k\neq m}(\bm{Tx})_{k}$, then the proposed algorithm determines $\tilde{m}$ such that $m \in  \mathbb{C}$.
\end{theorem}
\begin{proof}
Without loss of generality, we assume $m=0$, implying $\tilde{m}=0$ such that $m \in \left\{ k |  i*M \text{ for }  i=0,...,\lc \frac{N}{M} \rc -1  \right\}$.
If $\displaystyle (\tilde{\bm{y}})_{0} > \max_{i \neq 0} (\tilde{\bm{y}})_{i} $, the proposed algorithm correctly picks $\tilde{m}=0$.

To show the sufficient condition of $\displaystyle (\tilde{\bm{y}})_{0} > \max_{i \neq 0} (\tilde{\bm{y}})_{i} $, we first derive the lower bound of $(\tilde{\bm{y}})_{0}$, which is $$ (\tilde{\bm{y}})_{0} = (\bm{Tx})_{0}+ \sum_{i=1}^{\lc \frac{N}{M} \rc -1}(\bm{Tx})_{iM} $$ $$ \geq (\bm{Tx})_{0} - \left( \lc \frac{N}{M} \rc -1 \right)\max_{k\neq 0}(\bm{Tx})_{k}.$$
Then, the upper bound of $\max_{i \neq 0} (\tilde{\bm{y}})_{i} $ is derived to be $  \max_{i \neq 0} (\tilde{\bm{y}})_{i} \leq \lc \frac{N}{M} \rc \max_{k\neq 0} (\bm{Tx})_{k}.$
Thus, we have
\begin{eqnarray*}
\begin{array}{l}
 (\bm{Tx})_{0} - \left( \lc \frac{N}{M} \rc -1 \right)\max_{k\neq 0} (\bm{Tx})_{k} >  \lc \frac{N}{M} \rc \max_{k\neq 0} (\bm{Tx})_{k} \\ \\
\Rightarrow \frac{1}{ 2\lc \frac{N}{M} \rc +1 }(\bm{Tx})_{0} > \max_{k\neq 0} (\bm{Tx})_{k}.
\end{array}
\end{eqnarray*}
We complete the proof.
\end{proof}

Although Theorem \ref{thm: sufficient condition for perfect match} provides the sufficient condition, the successful probability is still unknown. In fact, the probability is related to the signal type of $x$. We will give a more detailed analysis taking the practical application as an example later.

\subsection{Best matching position}\label{ssec: step 4}
The operation ``Matching position finding'' yields a set of candidate positions, but the unique position still is unknown.
The following discussion is based on the prerequisite that the set of candidate positions includes the correct solution $m$.

To identify the correct solution, it is equivalent to solving an unknown variable $m$ such that $m \equiv \tilde{m} \text{ mod } M$.
The problem is efficiently solved by Chinese Remainder Theorem as follows.

\begin{theorem} \label{thm: CRT}
(Chinese Remainder Theorem (CRT)) Any integer $m$ is uniquely specified mod $N$ by its remainders modulo $\alpha$ relatively prime integers $M_{1}$, $M_{2}$, ..., $M_{\alpha}$ as long as $\prod^{\alpha}_{i=1}M_{i} \geq N$.
\end{theorem}

More specifically, we consider two co-prime integers in this paper.
Let $M_{1}$ and $M_{2}$ be relatively prime integers such that $M_{1}M_{2}\geq N$.
Hence, the two equations, $m \equiv \tilde{m}_{1} \text{ mod } M_{1}$ and $m \equiv \tilde{m}_{2} \text{ mod } M_{2}$, where $\tilde{m}_{1}$ and $\tilde{m}_{2}$ are the results obtained from the third operation ``Matching position finding,'' have the unique solution $m$ that is the best matching position.

\subsection{Practical Applications}\label{ssec: prac app}
In this section, we discuss the practical applications using the proposed algorithm.
Among them, synchronization is a critical issue in communications.
For example, global positioning system (GPS) consumes 30\%-75\% power for synchronization.
Thus, it is crucial to develop cost-effective GPS synchronization.
The problem is defined as follows.
\begin{enumerate}
\item[$\bullet$] In the sender, a spreading code $\bm{x} \in \left\{ 1, -1 \right\}^{N}$, which is also known by the receiver, is sent to the receiver.
\item[$\bullet$] The receiver obtains the delayed code: $$\bm{x}^{m}=\left[ (\bm{x})_{m},...,(\bm{x})_{N-1},(\bm{x})_{0},...,(\bm{x})_{m-1} \right].$$
 \item[$\bullet$] By comparing the spreading code ($\bm{t}=\bm{x}$ in this paper) and the delayed code, the receiver solves $\displaystyle \hat{m}=\arg\max_{0<k < N}\ Sim\left(\bm{x}^{k},\bm{t}\right)$.
\end{enumerate}

It should be noted that whatever the size of $\bm{t}$ is, the computation complexity of fast convolution by FFT is invariant.
In other words, even though the size of template is smaller than $N$, it cannot reduce the computation cost.

Nevertheless, the proposed method is not necessary to set $\bm{t}=\bm{x}$ for synchronization and can benefit from low-dimensional template.
In fact, it suffices to set $\bm{t} = \left[(\bm{x})_{0},...,(\bm{x})_{K-1} \right]$ with the size $K\leq M<N$.
We prove in the following theorem the successful probability and computation complexity of our method.

\begin{theorem} \label{thm: binary signal reconstruction}
Let $\alpha = \min(M_{1},M_{2})$ and let $\beta= \max(M_{1},M_{2})$.
If $\bm{x} \in \left\{ 1, -1 \right\}^{N}$, where $1$ and $-1$ have equal probability, and $ \bm{t} = [(x)_{m},...,(x)_{m+K-1}]$ for any $m$, then our proposed algorithm perfectly identifies the matching result $\hat{m}$  to be equal to the correct position $m$, {\em i.e.}, $\hat{m}=m$, with the successful probability being larger than $1-2\left( O(\alpha e^{-\frac{K}{8\lc \frac{N}{\alpha} \rc }})+ O( \frac{1}{K}\lc \frac{N}{\alpha} \rc)) \right)$ and computation costs of $O(N+\beta \log \beta)$ for additions and $O(\beta \log \beta)$ for multiplications.
\end{theorem}
\begin{proof}
The proposed algorithm can be roughly divided into two matching problems based on $M_{1}$ and $M_{2}$.
Only when both matching problems succeed, the proposed method is considered to be successful.
In the following, we first take the matching problem based on $M_{1}$ as an example.
Definitely, the analysis is also applied to another matching problem involving $M_2$.

The probability that the matching problem employing $M_{1}$ reports an incorrect result is equal to $$ P_{F}^{M_{1}} = \text{Pr}\left[  (\hat{\bm{T}}\bm{y})_{m} \leq \max_{i \neq m} (\hat{\bm{T}}\bm{y})_{i}  \right].$$
Since $(\bm{Tx})_{k}= \sum_{i=0}^{K-1} (\bm{t})_{i}(\bm{x})_{k+i} $, we have
\begin{enumerate}
\item[$\bullet$] If $k=m$, $(\bm{Tx})_{m}=\sum_{i=0}^{K-1} (\bm{x})_{m+i}^{2} = K $. Furthermore, $E\left[ (\bm{Tx})_{m} \right] = K$ and $Var\left[ (\bm{Tx})_{m} \right] = 0$.
\item[$\bullet$] If $k \neq m$, $E\left[(\bm{Tx})_{k} \right] = \sum_{i=0}^{K-1} E\left[ (\bm{t})_{i}(\bm{x})_{k+i}\right] = 0 $ due to the independence between $(\bm{t})_{i}$ and $(\bm{x})_{k+i}$. Furthermore, both $(\bm{t})_{i}(\bm{x})_{k+i}$ and $(\bm{t})_{j}(\bm{x})_{k+j}$ are independent for $i \neq j$. Thus,  $Var\left[(\bm{Tx})_{k} \right] = \sum_{i=0}^{K-1} Var\left[ (\bm{t})_{i}(\bm{x})_{k+i}\right] = K $.
\end{enumerate}

Recall that $\displaystyle (\hat{\bm{T}}\bm{y})_{k} = \sum_{i=0}^{\lc \frac{N}{M_{1}} \rc -1}(\bm{Tx})_{k+iM_{1}}$.\\ Let $S_{k} = \left\{ l | l = k+i*M_{1} \text{ for }  i=0,...,\lc \frac{N}{M_{1}} \rc -1   \right\}$. Thus,
\begin{enumerate}
\item[$\bullet$] If $m \in S_{k}$, $\displaystyle E\left[ (\hat{\bm{T}}\bm{y})_{k} \right] = \sum_{i=0}^{\lc \frac{N}{M_{1}} \rc -1}E\left[ (\bm{Tx})_{k+iM_{1}} \right] = K$. Furthermore, $E\left[ (\bm{Tx})_{k+iM_{1}} (\bm{Tx})_{k+jM_{1}}\right]=0$ for $i \neq j$. It implies $Cov\left[  (\bm{Tx})_{k+iM_{1}}, (\bm{Tx})_{k+jM_{1}}\right]=0$ for $i \neq j$. Thus, $\displaystyle Var\left[ (\hat{\bm{T}}\bm{y})_{k} \right] = \sum_{i=0}^{\lc \frac{N}{M_{1}} \rc -1}Var\left[ (\bm{Tx})_{k+iM_{1}} \right] + \sum_{i \neq j}^{\lc \frac{N}{M_{1}} \rc -1}Cov\left[ (\bm{Tx})_{k+iM_{1}},(\bm{Tx})_{k+jM_{1}} \right]=(\lc \frac{N}{M_{1}} \rc -1)K$.
\item[$\bullet$] If $m \not\in S_{k}$, $ E\left[ (\hat{T}y)_{k} \right] = 0 $ and $ Var\left[ (\hat{\bm{T}}\bm{y})_{k} \right]= \lc \frac{N}{M_{1}} \rc K$.
\end{enumerate}

Then, we will bound the probability $P_{F}^{M_{1}}$ by the following events: $E_{1}: \exists k \text{ and }  m \not\in S_{k}$, $\ s.t \ (\hat{\bm{T}}\bm{y})_{k} \geq \frac{K}{2}$; $E_{2}: \exists k \text{ and }  m \in S_{k} $, $\ s.t \ (\hat{\bm{T}}\bm{y})_{k} \leq \frac{K}{2}$. If none of the events holds, then the algorithm output is correct. In other words, $P_{F}^{M_{1}} = \text{Pr}\left[ E_{1} \right]+\text{Pr}\left[ E_{2} \right]$.

First, we discuss the probability, $\text{Pr}\left[ E_{1} \right]$. We start from the probability that for each $k$ such that $m \not\in S_{k}$, $(\hat{\bm{T}}\bm{y})_{k} \geq \frac{K}{2}$. It should be noted that, in this case, $(\hat{\bm{T}}\bm{y})_{k}$ is considered as a sum of independent random variables taking values in $\{ 1, -1 \}$ with the probability $\frac{1}{2}$. Thus, for each $k$, $\text{Pr}\left[ (\hat{\bm{T}}\bm{y})_{k} \geq \frac{K}{2} \right]$ is derived by using Chernoff bound as follows:
$$
\text{Pr}\left[ (\hat{\bm{T}}\bm{y})_{k} \geq \frac{K}{2} \right] \leq e^{ -\frac{ (\frac{K}{2})^{2}  }{2 \lc \frac{N}{M_{1}} \rc}K } =  e^{- \frac{ K  }{8 \lc \frac{N}{M_{1}} \rc} }.
$$
Since the cardinality of $\left\{k | m \not\in S_{k} \right\} $ is at most $M_{1}-1$, then
$$
\text{Pr}\left[ E_{1} \right]\leq (M_{1}-1)e^{- \frac{ K  }{8 \lc \frac{N}{M_{1}} \rc} } < M_{1}e^{- \frac{ K  }{8 \lc \frac{N}{M_{1}} \rc} }.
$$

Second, $\text{Pr}\left[ E_{2} \right]$ is bounded by using Chebyshev's inequality:
$$
\begin{array}{l}
\text{Pr}\left[ E_{2} \right] = \text{Pr}\left[ (\hat{\bm{T}}\bm{y})_{k} \leq \frac{K}{2} \right] =  \text{Pr}\left[ (\hat{\bm{T}}\bm{y})_{k} \leq K-\frac{K}{2} \right] \\
= \text{Pr}\left[ (\hat{\bm{T}}\bm{y})_{k} - K \leq -\frac{K}{2} \right] = \text{Pr}\left[  K - (\hat{\bm{T}}\bm{y})_{k} > \frac{K}{2} \right] \\
\leq \frac{Var\left[ (\hat{\bm{T}}\bm{y})_{k} \right]}{(\frac{K}{2})^{2}}=\frac{4(\lc \frac{N}{M_{1}} \rc -1)}{K}.
\end{array}
$$
Consequently, $P_{F}^{M_{1}} \leq \text{Pr}\left[ E_{1} \right] + \text{Pr}\left[ E_{2} \right] \leq  M_{1}e^{- \frac{ K  }{8 \lc \frac{N}{M_{1}} \rc} } + \frac{4(\lc \frac{N}{M_{1}} \rc -1)}{K}$.

Similarity, the above derivations apply to the second matching problem based on $M_{2}$.
It results in $P_{F}^{M_{2}} \leq M_{2}e^{- \frac{ K  }{8 \lc \frac{N}{M_{2}} \rc} } +  \frac{4(\lc \frac{N}{M_{2}} \rc -1)}{K}$.

In sum, the proposed algorithm succeeds with the probability being larger than $(1-P_{F}^{M_{1}})(1-P_{F}^{M_{2}})$. If $M_{1} > M_{2}$, then $P_{F}^{M_{1}} < P_{F}^{M_{2}}$.
For simplification, let $\alpha = \min(M_{1},M_{2})$ and we have $(1-P_{F}^{M_{1}})(1-P_{F}^{M_{2}}) \geq  (1-P_{F}^{\alpha})^{2} \geq  1-2\left( O(\alpha e^{-\frac{K}{8\lc \frac{N}{\alpha} \rc }})+ O( \frac{1}{K}\lc \frac{N}{\alpha} \rc)) \right)$.
We complete the proof regarding the successful probability.

To prove the computation complexity, we check Algorithm \ref{alg:FTM} in a step-by-step manner.
Step 2 depends on how to pick co-prime integers. In our case, co-prime integers are assigned by setting $M_{1}=K$ and $M_{2}=K+1$, and, thus, the cost is negligible.
Step 3 obviously costs $O(N)$ additions.
Step 4 assigns vectors with the sizes being smaller than $M_{1}$ or $M_{2}$ and the cost is also negligible. Steps 5 and 6 perform FFT and cost $O(M_{1} \log M_{1})+O(M_{2} \log M_{2})$.
Step 7 searches the maximum value within the vector with size $M_{1}$ or $M_{2}$, and the cost is negligible. Furthermore, in Steps 8 and 9, the cardinality of $S_{M_{1}}$ and $S_{M_{2}}$ are $\lc \frac{N}{M_{1}} \rc$ and $\lc \frac{N}{M_{2}} \rc$ respectively. Both of them are smaller than $O(N)$

Consequently, the computation cost is bounded by Step 3, and Steps 5 and 6 with the total cost being $O(N+M_{1} \log M_{1} + M_{2} \log M_{2})$ for additions and $O(M_{1} \log M_{1} + M_{2} \log M_{2})$ for multiplications.
For simplification, if we let $\beta= \max(M_{1},M_{2})$, then the cost is bounded by $O(N+\beta \log \beta)$ for additions and $O(\beta \log \beta)$ for multiplications.
We complete the proof.
\end{proof}


\section{Simulation Results}\label{sec:exp}

The simulations were conducted in an Matlab R2012b environment with an Intel CPU Q6600 and $16$ GB RAM under Microsoft Win7 ($64$ bits).
We compare the proposed algorithm implemented in C with fast convolution by FFT \cite{Lewis1995} and sparse FFT \cite{Hassanieh2012acm}.
Another goal is also to verify our theoretical analyses.

The testing procedure is:
\begin{enumerate}
\item Generate $\bm{x} \in \left\{ 1, -1 \right\}^{N}$ and template is extracted from $\bm{x}$ with $\bm{t}=\left[ (\bm{x})_{m},...,(\bm{x})_{m+K-1} \right]$.
\item Input $\bm{x}$ and $\bm{t}$ into Algorithm \ref{alg:FTM} and output $\hat{m}$, where $M_{1}=K$ and $M_{2}=K+1$.
\item If $m=\hat{m}$, the trial is successful. Otherwise, it is failed.
\end{enumerate}
We repeats this procedure $1000$ times and calculates the successful probability.

In Fig. \ref{fig:Computation Cost}, we verify the computation cost derived in Theorem \ref{thm: binary signal reconstruction}.
Fast convolution by FFT is considered as the baseline.
It should be noted that all the results shown in Fig. \ref{fig:Computation Cost} have successful probabilities fixed at $100\%$.
In Fig. \ref{fig:Computation Cost}(a), the results for our method are shown by fixing $N$ at $2^{26}$ with $K$ being increased from $2^{16}$ to $2^{25}$.
However, it should be noted that the complexities of \cite{Lewis1995} and \cite{Hassanieh2012acm} only depend on $N$.
Thus, their computation costs are invariant to $K$.
On the other hand, the proposed algorithm benefits from the smaller size of template.
One can observe that the proposed method outperforms  \cite{Lewis1995} and  \cite{Hassanieh2012acm} when $\frac{N}{K}\geq 2^{2}$ and $\frac{N}{K}\geq 2^{4}$, respectively, implying that our algorithm is efficient when the template size is (far) smaller than the corresponding signal size.
Furthermore, when $K \leq 2^{21}$, the computation cost is, in fact, dominated by $O(N)$ instead of $O(M \log M)$ and increased slowly.
On the contrary, when $O(M \log M)$ dominates the computation cost, which will be about twice larger when $K$ is doubled.

In Fig. \ref{fig:Computation Cost}(b), the results for our method are shown by fixing $K$ with $N$ being increased from $2^{16}$ to $2^{26}$.
In addition, the results for so-called ``Proposed Method (best)'' are shown, where $N$ still is increased from $2^{16}$ to $2^{26}$ but, for each $N$, $K$ is assigned as small as possible such that successful probabilities achieves $100\%$ based on Theorem \ref{thm: binary signal reconstruction}.
Overall, these results indicate that
(1) Our method outperforms \cite{Lewis1995} and \cite{Hassanieh2012acm} especially when $N$ is large enough, implying that, for a large-scale problem, it will be more efficient.
(2) When $N \leq 2^{21}$, small sizes of templates (for dot-curve) suffice to quickly achieve $100\%$ successful template matching.
In a similar way to Fig. \ref{fig:Computation Cost}(a), when $N \geq 2^{21}$, $O(N)$ overwhelms $O(M\log M)$.
In this case, the slopes of red solid and dot curves (our methods) approximate to those of solid black (FFTW) and solid blue (sFFT) curves, where both complexities are $O(N)$ and $O(N \log N)$, respectively.

\begin{figure}[h]
\begin{minipage}[b]{.99\linewidth}
  \centering{\epsfig{figure=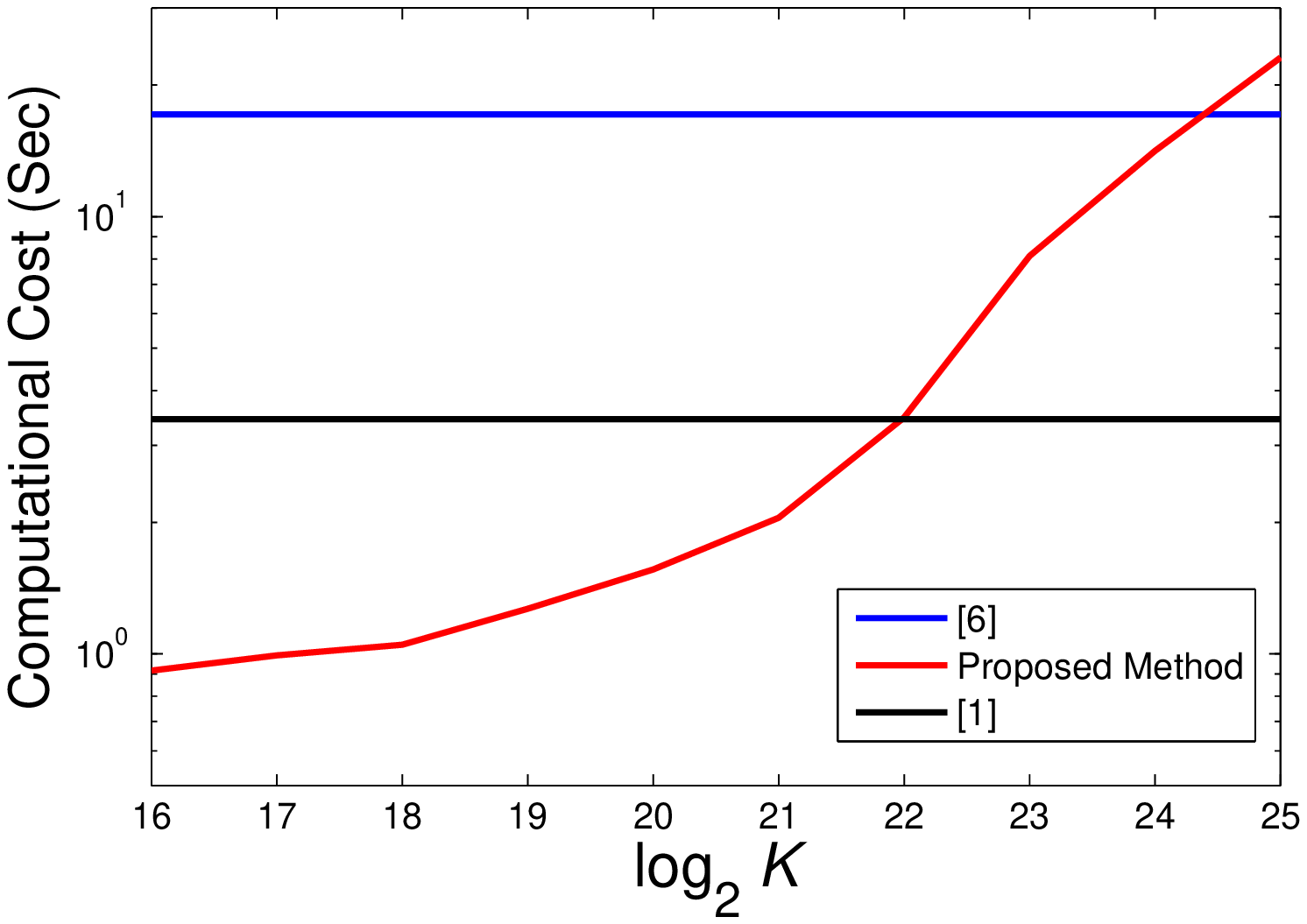,width=2.5in}}
  \centerline{\quad (a)}
\end{minipage}
\begin{minipage}[b]{.99\linewidth}
  \centering{\epsfig{figure=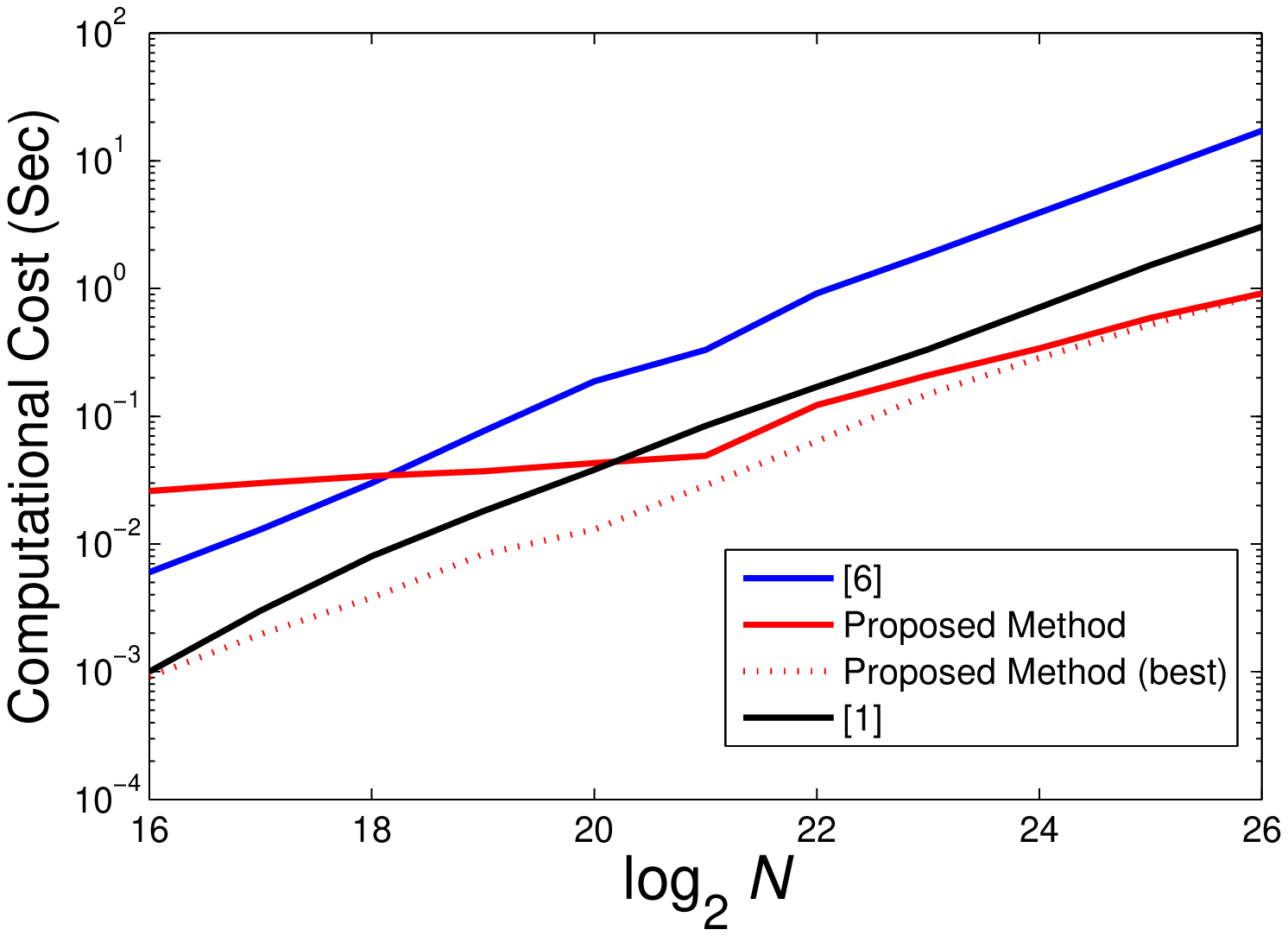,width=2.5in}}
  \centerline{\quad (b)}
\end{minipage}
\caption{(a) The computation cost versus different $K$'s under $N=2^{26}$. (b) The computation cost versus different $N$'s under $K=2^{16}$. All comparisons were conducted with successful probabilities being fixed at $100\%$}
\label{fig:Computation Cost}
\end{figure}

Table \ref{table: succ prob fixed ratio} further verifies the theoretical probability derived in Theorem \ref{thm: binary signal reconstruction}.
By Theorem \ref{thm: binary signal reconstruction}, the proposed algorithm succeeds with the probability being larger than $1-2\left( O(\alpha e^{-\frac{K}{8\lc \frac{N}{\alpha} \rc }})+ O( \frac{1}{K}\lc \frac{N}{\alpha} \rc)) \right)$. In our case, $\alpha = \min(M_{1},M_{2})$, where $M_{1}$ and $M_{2}$, in fact, were set to $K$ and $K+1$, respectively.
Thus, when the condition $K^{2}\geq O(N)$ holds, $\frac{1}{K}\lc \frac{N}{\alpha} \rc = \frac{1}{K}\lc \frac{N}{K} \rc \rightarrow 0$.
Similarly, we have $\alpha e^{-\frac{K}{8\lc \frac{N}{\alpha} \rc }} \rightarrow 0 $ with $\frac{K^{2}}{\log K} \geq O(N) $.
In this experiment, $\frac{N}{K}$ was fixed as $2^{10}$.
The results show that along with the increase of $N$, the condition $\frac{K^{2}}{\log K}=\frac{N^{2}}{2^{20}\log N - c} \geq O(N)$ ($c$ is a positive constant), which is equivalent to $\frac{N}{\log N}\geq O(2^{20})$, holds.
Under the condition, Algorithm \ref{alg:FTM} succeeds with high probability, as depicted in Table  \ref{table: succ prob fixed ratio}.

\begin{table}[!htbp]
\caption{Successful probabilities under $K=\frac{N}{2^{10}}$.}
\centering
\begin{tabular}{|c|c|c|c|c|c|}
\hline
$N$  & $2^{23}$ & $2^{24}$ & $2^{25}$ & $2^{26}$ & $2^{27}$\\
\hline
Successful Prob. & 0.05 & 0.21 & 0.85 & 1 & 1\\
\hline
\end{tabular}
\label{table: succ prob fixed ratio}
\end{table}

Finally, we test whether the proposed method is robust to noisy inference.
In this case, let $\bm{x}_{e}=x+\bm{e}$, where $\bm{e} \in \mathbb{R}^{N}$ is additive Gaussian random noise.
Both $x_{e}$ and $t$ were fed into Algorithm \ref{alg:FTM}.
The parameters, $N=2^{26}$ and $M=2^{16}$, were chosen because the corresponding successful probability lives on edge between $100\%$ and $<100\%$, which is expected to be interfered by noise obviously.
One can observe from Table \ref{table: robust under noise} that our method works well when SNRs are larger than $1$ dB.

\begin{table}[!htbp]
\caption{Successful probabilities under $N=2^{26}$, $K=2^{16}$, and different SNRs.}
\centering
\begin{tabular}{|c|c|c|c|c|c|c|}
\hline
SNR (dB) & $20$ & $10$ & $6$ & $1$ & $-2$ &$-6$\\
\hline
Successful Prob. & 1 & 1 & 0.95 & 0.72 & 0.48 & 0.06\\
\hline
\end{tabular}
\label{table: robust under noise}
\end{table}

\section{Conclusions and Future Work}\label{sec:conclu}
We present a fast and cost-effective template matching scheme in this paper.
We exploit the commutative property of partial circulant matrix to design the sensing matrix for template matching.
Our theoretical analyses and simulation results show that the proposed method outperforms FFT- and sparse FFT-based methods.
The future work will be examining fast template matching with similarity measures other than cross-correlation.

\section{Acknowledgment}
This work was supported by Ministry of Science and Technology, Taiwan (ROC), under grants MOST 104-2221-E-001-019-MY3 and 104-2221-E-001-030-MY3.

\bibliographystyle{IEEEbib}	
\bibliography{refs}		

\end{document}